\documentclass[aps,prl,
twocolumn,superscriptaddress]{revtex4}
 
\usepackage{graphicx}
\usepackage{enumerate,amsthm,amssymb,color}
\usepackage{bbm}
\usepackage{bm}
\usepackage{amsmath}
\usepackage{amsfonts}
\usepackage{latexsym}
\usepackage{epstopdf}
\usepackage{graphicx}
\usepackage{hyperref}
\usepackage{paralist}

\newtheorem{theorem}{Theorem}

\newtheorem{definition}{Definition}

\usepackage{color,graphicx}

\newcommand{\spa}[1]{\mathcal{#1}}
\newcommand{\Tr}{\mathrm{Tr}}
\newcommand{\half}{\frac{1}{2}}

\newcommand{\calA}{\mathcal{A}}
\newcommand{\calB}{\mathcal{B}}

\newcommand{\tr}{\mathrm{Tr}}
\newcommand{\ket}[1]{| #1 \rangle}
\newcommand{\bra}[1]{\langle #1 |}

\newcommand{\ketbra}[2]{\ket{#1} \bra{#2}}
\newcommand{\kb}[1]{\ketbra{#1}{#1}}
\newcommand{\braket}[2]{\langle #1 | #2 \rangle}
\newcommand{\dsum}{\displaystyle\sum}
\newcommand{\inner}[2]{\langle #1, #2 \rangle}

\newcommand{\set}[1]{\left\{ #1 \right\}}

\newcommand{\norm}[1]{\left\| #1 \right\|}

\newcommand{\zo}{\{0,1\}}

\newcommand{\OT}{\mathrm{OT}}

\newcommand{\BC}{\mathrm{BC}}
\newcommand{\CF}{\mathrm{CF}}

\newcommand{\CHSH}{\mathrm{CHSH}}

\newcommand{\comment}[1]{}
\newcommand{\COMMENT}[1]{}

\begin{document}

\date{\today}
\author{Jamie Sikora} 
\affiliation{Laboratoire d'Informatique Algorithmique: Fondements et Applications, CNRS - Universit\'e Paris Diderot, France.}
\author{Andr\'e Chailloux} \affiliation{
INRIA Paris-Roquencourt, SECRET Project-Team, France.}
\author{Iordanis Kerenidis$^{1,}$} 
\affiliation{
Centre for Quantum Technologies, National University of Singapore, Singapore.}

\title{Strong connections between quantum encodings, non-locality and quantum cryptography}

\begin{abstract}
Encoding information in quantum systems can offer surprising advantages but at the same time there are limitations that arise from the fact that measuring an observable may disturb the state of the quantum system. In our work, we provide an in-depth analysis of a simple question: What happens when we perform two measurements sequentially on the same quantum system? This question touches upon some fundamental properties of quantum mechanics, namely the uncertainty principle and the complementarity of quantum measurements. Our results have interesting consequences, for example they can provide a simple proof of the optimal quantum strategy in the famous Clauser-Horne-Shimony-Holt game. 
Moreover, we show that the way information is encoded in quantum systems can provide a different perspective in understanding other fundamental aspects of quantum information, like non-locality and quantum cryptography. We prove some strong equivalences between these notions and provide a number of applications in all areas.
\end{abstract}

\maketitle

Quantum information studies how information is encoded in quantum systems and how it can be observed through measurements. On one hand, the exponential number of amplitudes that describe the state of a quantum system can be used in order to encode a vast amount of classical information into the state of a quantum system. Hence, we can use quantum information to resolve many distributed tasks much more efficiently than with classical information \cite{Raz99,BCWdW01,GavinskyKKRW08}. On the other hand, quantum information does not always offer advantages, since every time an observer measures a quantum system its state may collapse and information may become irretrievable. For example, Holevo's theorem \cite{Hol73}, asserts that one quantum bit can be used to transmit only one bit of classical information and no more. 

The intricate interplay between encoding information in quantum systems and measurement interference is at the heart of some fundamental results in quantum information, from Bell inequalities~\cite{Bell64} to quantum key distribution \cite{BB84}.
Our goal is to deepen our understanding of the connections between quantum encodings, non-locality, and quantum cryptography and provide new insight on the power and limitations of quantum information, by looking at it through these various lenses. 
 
This paper links three seemingly unrelated
concepts in quantum information (encodings, non-local games, and cryptographic primitives) via properties of sequential non-commuting measurements. The technical part of this paper examines quantum encodings and bounds the success of sequentially measuring an encoding of two bits (or strings) to learn their XOR. We then show how these bounds can be used to study not only encodings, but non-local games and cryptographic tasks as well. The conceptual part of this paper discusses how the applications we consider are all equivalent in some sense. When viewing each as extracting information from a quantum encoding, we are able to preserve the three notions: (1) hiding the XOR in the encoding, (2) providing perfect security in the cryptographic task, and (3) satisfying the non-signaling principle in the non-local game. 

In addition to providing philosophical insights towards each of these quantum tasks, we combine the technical and conceptual tools in this paper to give applications in all areas. 

\section{Quantum encodings and complementarity of measurements}
One of the fundamental postulates of quantum mechanics is Heisenberg's uncertainty principle which shows that it is impossible to perfectly ascertain the momentum and position of a particle. More precisely, entropic uncertainty relations provide explicit bounds on the entropy of the outcome distributions of the different measurements. For example, if we consider two measurements in the computational and Hadamard bases, then no matter the state of the quantum system, there is always some entropy in at least one of the outcome distributions, hence the measurement outcomes cannot be perfectly predicted simultaneously. 

Another important notion, which is more closely related to quantum encodings, is the complementarity of quantum measurements. Complementarity analyzes what happens to the outcome distributions of measurements when performed sequentially on the same system. We say that two measurements are perfectly complementary, if after having performed the first measurement, no more information can be extracted by performing the second measurement on the post-measured state. This is, for example, the case with a Hadamard and a computational basis measurement, or any measurement after a complete projective measurement. On the other hand, they are non-complementary if after measuring with one, the outcome distribution of the second is unaffected. 

We make the connection of complementarity and quantum encodings clearer by considering the following  scenario: Let us consider two different observables that take binary values $x_0 \in \{0,1\}$ and $x_1 \in \{0,1\}$ according to some known distribution. Assume that given one copy of a quantum system (of any dimension) in state $\rho_{x_0,x_1}$, i.e., a quantum encoding of the bits $x_0,x_1$, there exists a quantum measurement, i.e., a decoding procedure, that correctly measures $x_0$ with probability $p_0$ and a different measurement that correctly measures $x_1$ with probability $p_1$. We would like to analyze these probabilities and more specifically the {\em average decoding probability} $(p_0+p_1)/2$.  

Uncertainty relations show that when the measurements are ``incompatible'' the average decoding probability cannot be too large. For example, for the computational and Hadamard bases one can show this probability is always at most $\cos^2(\pi/8)$. 
There are many cases where we do not know the different measurement operators, only the probabilities they succeed. For example, one may not know the measurements used in an implicit strategy in a cryptographic protocol or quantum non-local game where the only defining property of the strategy is the success probability. Could we still provide some interesting bound on the average decoding probability that would hold independent of the measurement operators, possibly by relating it to some other property of the quantum encoding? 

We provide such bounds by relating the average decoding probability to the decoding probability of some other function $f(x_0,x_1)$ of the bits. 
Classically, it is straightforward to relate the probability of decoding $f(x_0,x_1)$ to the probabilities of decoding each bit $x_i$; in the quantum world, this task is delicate. 
Suppose we want to compute the XOR of the two bits (i.e., compute whether the two bits have the same value or not), and for this we perform the measurement for each bit $x_i$ in sequence. Once the first bit is decoded, the post-measured state is an eigenstate of the first operator, hence the probability of then correctly decoding the second bit may have changed. 

Much of the previous literature about measuring the post-measured state concerns ideas surrounding Heisenburg's uncertainty principle (see, for example,~\cite{O03} and the references therein). In a setting more related to this paper, post-measurement information has been used for state discrimination~\cite{BWW08, GW10}. This is useful for cryptography in the bounded-storage model~\cite{DFSS08} and the noisy-storage model~\cite{WST08, S10}. 

\section{Learning relations} 

Our first contribution is an analysis of the process of sequentially performing two measurements on the same quantum state: Let $\ket{\psi}$ be a pure state and $\{ C, 1-C \}$, $\{ D, 1-D \}$ be two projective measurements such that $\cos^2(\alpha) := \norm{C \ket{\psi}}^2_2 \geq \half$ and ${\cos^2(\beta) := \norm{D \ket{\psi}}^2_2 \geq \half}$, where $C$ and $D$ correspond to correctly measuring. Through geometric arguments, we bound the probability that both measurements succeed (give the correct guess) or both fail (give the incorrect guess) as:
\vspace{-0.1cm}
\begin{eqnarray}
 \cos^2(\alpha - \beta) \geq \norm{C D \ket{\psi} }_2^2 + \norm{(1-C)(1-D) \ket{\psi} }_2^2 \nonumber\\
\geq \cos^2(\alpha + \beta). \label{claim:technical} 
\end{eqnarray} 

\vspace{-0.5cm}

\begin{figure}[htbp] 
   \centering
   \includegraphics[width=3.25in]{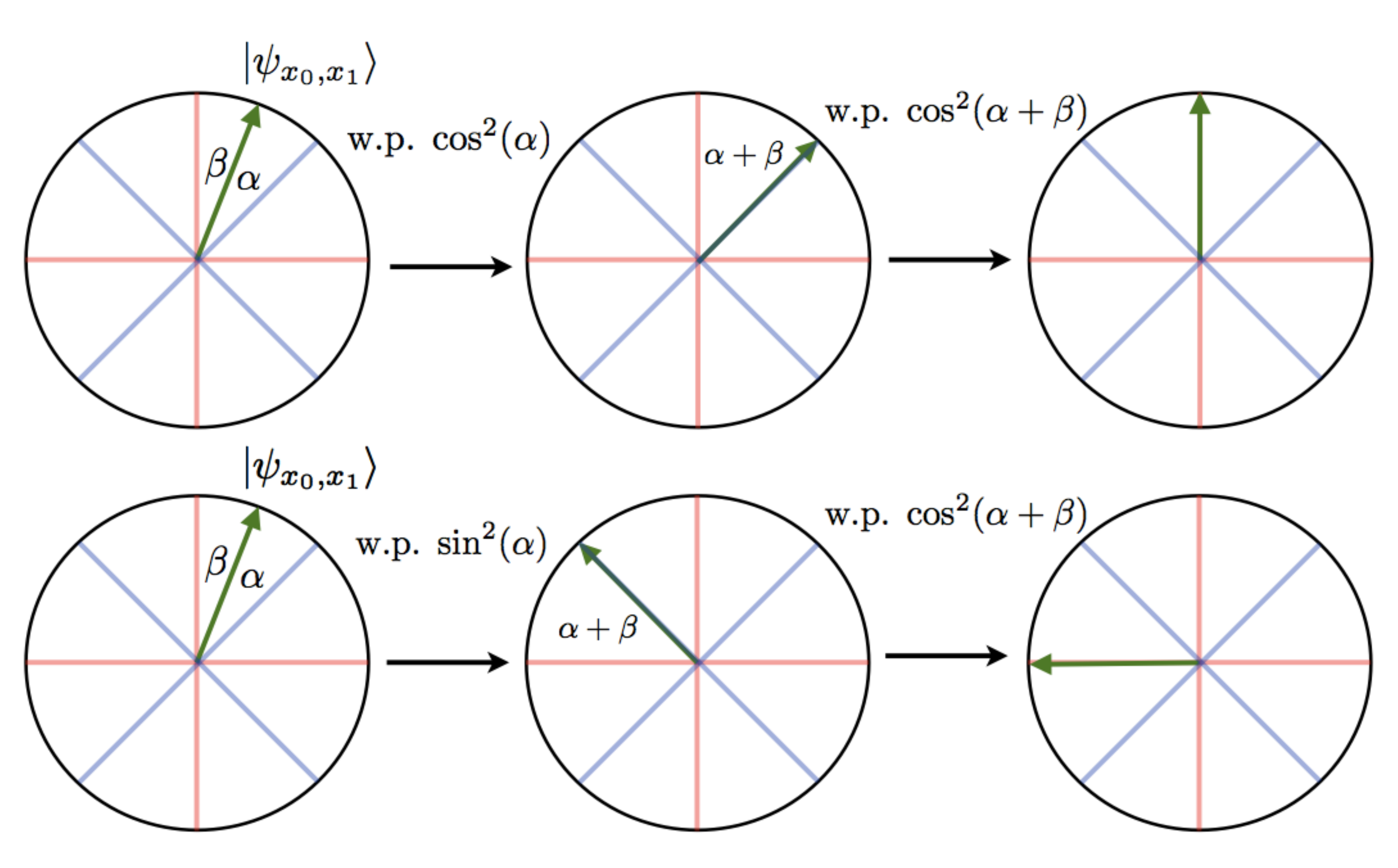} 
   \caption{(Color online) Simple scenario for the lower bound in Equation~(\ref{claim:technical}): Two-outcome projective measurements on a pure state in two dimensions. By successive measurements, one can learn the XOR by learning both bits correctly or by learning both bits incorrectly. This occurs with probability $\cos^2(\alpha) \cos^2(\alpha+ \beta) + \sin^2(\alpha) \cos^2(\alpha+\beta) = \cos^2(\alpha + \beta)$.}
   \label{fig:example}
\end{figure}

In the language of quantum encodings, we can use~(\ref{claim:technical})  
to provide the following learning relation for bits, and extend it to strings. (The proof of Equation~(\ref{claim:technical})  and Theorem~\ref{LearningLemmata}, below, can be found in the appendix.) 

\begin{theorem} \label{LearningLemmata}
For any quantum encoding of bits $x_0$ and $x_1$, $\Pr[\textup{learning } x_0 \oplus x_1] \geq (2c-1)^2$, 
where we define ${c := \frac{1}{2} \Pr[\textup{learning } x_0] + \frac{1}{2} \Pr[\textup{learning } x_1]}$. For $x_0,x_1 \in \{0,1\}^n$, if $c \geq 1/2$, then we have 
${\Pr[\textup{learning } x_0 \oplus x_1] \geq \Pr[\textup{learning } (x_0, x_1)] \geq c(2c-1)^2}$.
\end{theorem} 
 
\vspace{-0.05cm}
 
The probability of learning a bit (or a bit string) is the maximum over all quantum measurements of correctly measuring the bit (or bit string). Theorem 1 shows that, independent of the measurements, the average probability of correctly measuring two observables cannot be very large unless at the same time the probability of correctly measuring both or none of the observables is large as well. A similar result has been obtained for a restricted class of encodings, those based on \emph{hyperbits}~\cite{PW12}.

We can now define a measure of complementarity $\Gamma$, as the difference between the probability of decoding the XOR of the two bits and the probability had the measurements been non-complementary. By Equation~(\ref{claim:technical}),  
\vspace{-0.0cm}
\begin{eqnarray} 
|\Gamma |
& = &
\left| \norm{C D \ket{\psi} }_2^2 + \norm{(1-C)(1-D) \ket{\psi} }_2^2 \right. \nonumber \\ 
& - & 
\left. \norm{C \ket{\psi} }_2^2 \norm{D \ket{\psi} }_2^2 - \norm{(1-C)  \ket{\psi} }_2^2 \norm{(1-D) \ket{\psi} }_2^2 \right| \nonumber \\\vspace{-0.1cm}
& \leq & \half \sin(2\beta)\sin(2\alpha). 
\end{eqnarray}
Note $\Gamma$ is zero for non-complementary measurements and our bound can be saturated, e.g., when $C = D$ we have $\Gamma = \half \sin(2\beta)\sin(2\alpha)$, and for $C = 1-D$ we have ${\Gamma = - \half \sin(2\beta)\sin(2\alpha)}$. 

 
\section{The Clauser-Horne-Shimony-Holt game as a quantum encoding} 
Non-locality is a fundamental property of quantum information. Here, two space-like separated parties, Alice and Bob, initially share some resource and do not communicate further. We study the joint probability distributions of measurement outcomes that can arise when Alice and Bob perform measurements locally. Bell inequalities provide bounds on the possible distributions when the initial resource is classical and we are interested in the maximum violation when Alice and Bob share quantum entanglement. 

One can describe Bell inequalities as games between Alice and Bob. For example, in the Clauser-Horne-Shimony-Holt $(\CHSH)$ game~\cite{CHSH69}, Alice receives a random $x \in \zo$ and outputs $a \in \zo$ and Bob receives a random $y \in \zo$ and outputs $b \in \zo$. The quantum value of the game, $\omega^*(\CHSH)$, is the maximum probability that $a \oplus b = yx$ over all initial states and all measurement operators. There is a quantum strategy to win this game with probability $\cos^2(\pi/8)$; moreover, Tsirelson's bound shows this value is optimal \cite{Tsi87}. 

Recently, non-locality has been studied from the point of view of information. The goal is to understand quantum mechanics through information principles: for example, why is there a quantum strategy for the CHSH game with probability exactly $\cos^2(\pi/8)$ and not more? \emph{Information causality}, one such postulate about information transmission, asserts that any theory that abides to it must comply with Tsirelson's bound for the CHSH game \cite{PPKSWZ09}. 
To make the connection between non-locality and quantum information more clear, let us see how we can recast the CHSH game as a quantum encoding: Once Alice receives $x$ and measures $a$, Bob's post-measurement state can be seen as an encoding of $a$ and $x$. When $y=0$, Bob needs to output $a$ and when $y=1$, he needs to output $a \oplus x$. Hence, we can write the value as 
$ \omega^*(\CHSH) = \dfrac{1}{2}(\Pr[\text{Bob learns } a] + \Pr[\text{Bob learns } a \oplus x]). $\\
Note that the non-signaling condition of CHSH implies the probability of Bob guessing Alice's input $x$ is $1/2$ or equivalently the probability of learning the XOR of $a$ and $a \oplus x$ is $1/2$ (in this case, we say that the encoding ``hides'' the XOR). With this perspective, Theorem~\ref{LearningLemmata} provides an alternative proof of Tsirelson's bound, since solving the inequality $(2\omega^*(\CHSH)-1)^2 \leq 1/2$ gives  $\omega^*(\CHSH) \leq \cos^2(\pi/8)$.
 
\section{Learning relations and oblivious transfer} 
Another area where quantum information has had great impact is cryptography. The properties of quantum information, for example, the uncertainty principle, enable secure key distribution protocols \cite{BB84}, however, when the two parties do not trust each other, there are only partial advantages. For example, quantum protocols for coin flipping or bit commitment can only restrict cheating to a probability of $1/\sqrt{2}$ or $0.739$, respectively \cite{Kit03, CK09, CK11}. 
We wish to relate the ability to perform cryptographic primitives to non-locality and quantum encodings. 
 
We look at \emph{oblivious transfer} (OT), defined below. 

\begin{definition}[Imperfect oblivious transfer]
A quantum oblivious transfer protocol with correctness $p$, denoted here as $\OT_p$, is an interactive protocol with no inputs, between Alice and Bob such that: \\ 

\begin{compactitem}
\item Alice outputs two independent, uniformly random bits $(z_0,z_1)$ or Abort and Bob outputs uniformly random bit $b$ and another bit $w$ or Abort. 
\item If Alice and Bob are honest, $w = z_b$ with {probability $p$}. 
\item Alice and Bob can abort only if cheating is detected. 
\item If $p = 1$ we say the protocol is \emph{perfect}. \\ 
\end{compactitem}

Ideally at the end of the protocol, Bob should only learn the value of $z_b$ and Alice should remain oblivious to which bit Bob learned~\textup{\cite{W70, R81}}. 
\end{definition} 

We also examine quantum oblivious \emph{string} transfer protocols with correctness $p$, denoted here as $\OT_p^n$ which is defined analogously to an imperfect oblivious transfer protocol except $z_0$ and $z_1$ are $n$-bit strings. 

Oblivious transfer is the most important task in providing security between distrustful parties, since any complex operation can be rendered secure using secure oblivious transfer~\cite{Kil88}. Using Theorem~\ref{LearningLemmata}, we prove a series of new results for oblivious transfer~\footnote{We use a non-composable definition of security which makes our impossibility results even stronger.}. First, we extend the oblivious transfer bounds in \cite{CKS13} to oblivious {string} transfer, and show that in any protocol, either Alice can learn Bob's index or Bob can learn both of Alice's strings with probability at least $58.52 \%$ (proof in the appendix). Second, we consider the case when cheating Bob wants to learn the XOR of Alice's bits. Note that most definitions enforce that Bob gets no information about Alice's other bit (instead of the XOR of her bits). Classically, the two definitions are equivalent~\cite{DFSS06}. Quantumly, we use the XOR definition that relates directly to the CHSH game (discussed in the next section).   

\begin{theorem} \label{thm:LBCurve} 
For any $\OT_p$ protocol, we have \\
$p \leq \Pr[\textup{Alice learns b}] \left( \sqrt{\Pr[\textup{Bob learns } z_0 \oplus z_1]} + 1 \right)$. 
\end{theorem} 

\begin{proof} 
We show how to use oblivious transfer to  construct a \emph{coin flipping} protocol.   
A quantum coin flipping protocol with correctness $p$, denoted $\CF_p$, is an interactive protocol with no inputs, between Alice and Bob such that: \\

\begin{compactitem}
\item The protocol is aborted with probability $1-p$ when Alice and Bob are honest. 
\item If the protocol is not aborted, then they both output a randomly generated bit $c$. \\ 
\end{compactitem}

We say that the coin flipping protocol has cheating probabilities $A_{\CF}$ and $B_{\CF}$ where \\
\begin{compactitem}
\item $A_{\CF} := \max_{c \in \zo} \Pr[\textup{Bob accepts outcome } c]$, 
\item $B_{\CF} := \max_{c \in \zo} \Pr[\textup{Alice accepts outcome } c]$. \\ 
\end{compactitem}

The coin flipping protocol is as follows. \\ 

\begin{compactenum}
\item Alice and Bob perform the $\OT_p$ protocol so they have outputs $(z_0, z_1)$ and $(b, w)$ respectively. 
\item If no one aborted, then Alice sends randomly chosen $d \in_R \set{0,1}$ to Bob. 
\item Bob sends $b$ and $w$ to Alice. 
\item If $z_b$ from Bob is inconsistent with Alice's bits then Alice aborts. Otherwise, they both output $c = b \oplus d$. \\ 
\end{compactenum}

We see that when Alice and Bob are honest, Alice aborts in this protocol with probability $1-p$, since $p$ is the probability that Bob receives the correct bit in the $\OT_p$ protocol. If Alice does not abort, the outcome of the coin flipping protocol is random. 

\paragraph{Cheating Alice:} 
Let $A_{\OT}$ denote the probability Alice can learn $b$ in the $\OT_p$ protocol (without Bob aborting) and let $A_{\CF}$ denote the probability Alice can force honest Bob to accept a desired outcome in the coin flipping protocol. It is straightforward to see that $A_{\OT} = A_{\CF}$. 

\paragraph{Cheating Bob:} 
Let $B_{\OT}$ denote the probability Bob can learn $z_0 \oplus z_1$ in the $\OT_p$ protocol (without Alice aborting) and let $B_{\CF}$ denote the probability Bob can force honest Alice to accept a desired outcome in the coin flipping protocol.
Using our XOR learning relation for bits, and an analysis similar to the one in~\cite{CKS13}, we can show that $\dfrac{\sqrt{B_{\OT}} + 1}{2} \geq B_{\CF}$. 
Kitaev's lower bound for coin flipping~\cite{Kit03} states that 
\[ A_{\CF} B_{\CF} \geq \Pr[\textup{Alice and Bob honestly output } 0] \] 
for any quantum coin flipping protocol.   
In the case of the coin flipping protocol above, we have that Alice and Bob both output either bit with probability $p/2$ (since the protocol is aborted with probability $1-p$). Therefore, we have $A_{\OT} \dfrac{\sqrt{B_{\OT}} + 1}{2} \geq A_{\CF} B_{\CF} \geq \dfrac{p}{2}$  
implying 
$A_{\OT} \left( {\sqrt{B_{\OT}} + 1} \right) \geq p$, 
proving Theorem~2.  
\end{proof} 

Notice that for secure protocols with ${\Pr[\text{Alice learns b}] = \half \text{ and } \Pr[\text{Bob learns } z_0 \oplus z_1] = \half}$, we have $p \leq \cos^2(\pi/8)$, which shows that the secure oblivious transfer protocol in~\cite{BBBW83} is optimal. 
Last, by relating oblivious transfer and bit commitment protocols~\cite{CKS13}, we prove that in any $\OT$ protocol with $p=1$, Alice can learn Bob's index or Bob can learn the XOR of Alice's bits with probability at least $59.9 \%$ (proof in the appendix). 


\section{Equivalences between CHSH-type games, secure oblivious transfer and quantum encodings} 
So far,  we have used Theorem~\ref{LearningLemmata} to provide results about the CHSH game and oblivious transfer. 
We now show that these applications are deeply connected and can be extended to more intricate non-local games and oblivious transfer variants. Such non-local games are important since knowing their Bell inequality violations  brings us that much closer to understanding the true power of quantum entanglement and the hope of characterizing it as a resource via the right information postulate(s). 

We now consider \emph{secure} $\OT_p^n$ protocols where Alice can obtain no information about Bob's index $b$ (without him aborting) and Bob can obtain no information about $z_0 \oplus z_1$ (without Alice aborting). 

We also consider the following generalization of the CHSH game. 

\begin{definition}[CHSH$_{n}$ game] 
The \emph{$\CHSH_n$ game} is a game between Alice and Bob where: \\

\begin{compactitem}
\item Alice and Bob are allowed to create and share an entangled state $\ket{\psi}$ before the game starts. Once the game starts, there is no further communication between Alice and Bob. 
\item Alice receives a random string $x \in \zo^n$ and Bob receives a random bit $y~\in~\zo$. 
\item Alice outputs $a \in \zo^n$ and Bob outputs ${b \in \zo^n}$.  
\item Alice and Bob win if $a_i \oplus b_i = y \, x_i$, for all $i \in \{ 1, \ldots, n \}$. \\ 
\end{compactitem}
  
The value of the game, $\omega^*(\CHSH_n)$, is the maximum probability which Alice and Bob can win. 
\end{definition} 

The $\CHSH$ game is the special case when $n=1$ (we omit the subscript $1$ in this case). 

A relationship between learning probabilities and quantum games is pointed out in~\cite{OW10}, where they show that in any physical theory, the amount of non-locality and uncertainty of the theory are tightly linked. In our equivalences, we strengthen the quantum connection by conserving the notions of security / non-signaling / hidden XOR and we deal with the interactivity of oblivious transfer protocols. 

\begin{theorem} \label{thm:equivalence1}
The following four statements are equivalent for every $n \in \mathbb{N}$: \\ 
\begin{compactenum}
\item There is a quantum encoding of ${x_0, x_1 \in \zo^n}$ that hides the XOR and 
$\frac{1}{2} \sum_{c \in \zo} \Pr[\textup{learn } x_c] = p$.
\item There is a secure, non-interactive $\OT_p^{n}$ protocol.
\item There is a secure $\OT_p^{n}$ protocol.
\item There is a strategy for winning the game $\CHSH_n$ with probability $p$. \\ 
\end{compactenum} 
\end{theorem} 

\vspace{-0.5cm}

\begin{proof} 
We provide four reductions. \\ 

\noindent 
{\em (1.}$\implies${\em 2.)}.  
Let $\{ \rho_{x_0, x_1} : x_0, x_1 \in \zo^n \}$ be a set of quantum states and  $\{ \pi_{x_0, x_1} : x_0, x_1 \in \zo^n \}$ be a probability distribution satisfying the properties of statement 1 of Theorem~3. Alice chooses $x_0, x_1$ with probability $\pi_{x_0, x_1}$ and sends $\rho_{x_0, x_1}$ to Bob. Alice outputs \\
$(z_0, z_1) := ((1-a) x_0 + a x_1 + d_1, (1-a) x_1 + a x_0 + d_2), $
for random choices of $a \in \{0,1\}$ and $d_1,d_2 \in \{0,1\}^n$ that she sends to Bob. 
The first bit randomizes the success probabilities for Bob (so he has an equal probability of learning $z_0$ and $z_1$) and the $d_1, d_2$ bit strings ensure that Alice's outcomes are random. 
Bob picks a random bit $b$ and measures to learn $z_b$ depending on $a, d_1, d_2$. In particular, the probability of learning $z_b$ for $b \in \{ 0,1 \}$ is equal to the average decoding probability of $x_0$ and $x_1$, hence equal to $p$. Note that $z_0 \oplus z_1 = x_0 \oplus x_1 \oplus d_1 \oplus d_2$ is hidden from Bob and Alice cannot learn $b$ (since Bob does not send any message), thus this protocol is secure. 

\noindent \\ 
{\em (2.}$\implies${\em 4.)}. 
Suppose there is a secure, non-interactive $\OT_p^n$ protocol. Without loss of generality~\footnote{For our purposes, we can assume Alice discards her quantum state except for the registers containing $z_0$ and $z_1$.}, Alice and Bob's joint state from the non-interactive $\OT_p^n$ protocol is
$ 1/2^n \sum_{z_0, z_1 \in \zo^n} \kb{z_0, z_1} \otimes \rho_{z_0, z_1}, $
for some $\rho_{z_0, z_1}$ in Bob's space $\spa{B}$. Since Alice has no information about $b$, Bob can use $\rho_{z_0, z_1}$ and measurements $\{ M^0_{z_0} \}_{z_0 \in \zo^n},$ $\{ M^1_{z_1} \}_{z_1 \in \zo^n}$ to learn the value of Alice's first and second string, respectively, with
$\Pr[\textup{Bob learns } z_0] = \Pr[\textup{Bob learns } z_1]~=~p.
$
Consider some purification $\ket{\psi_{z_0,z_1}} \in \calA \otimes \calB$ of $\rho_{z_0,z_1}$ where $\calA$ is controlled by Alice. Let \\
${\ket{\Omega} := \dfrac{1}{2^n} \dsum_{z_0, z_1 \in \zo^n}  \ket{z_0 \oplus z_1}_{\spa{A}_1} \ket{z_0}_{\spa{A}_2} \ket{z_1}_{\spa{A}_3} \ket{\psi_{z_0, z_1}}_{\spa{AB}}}$, 
$\ket{\Omega_x}$ to be the post-measured state assuming Alice measured $\spa{A}_1$ to get $x$, and $\rho_x := \tr_{\spa{A}_2 \spa{A}_3 \spa{A}} \kb{\Omega_x}$ to be Bob's state.  We have  $\rho_x = \rho_{0},\; \forall x \in \zo^n,$ since Bob has no information about $z_0 \oplus z_1$. By Uhlmann's theorem, for all $x \in \zo^n$, there exists unitary $U_x$ on $\spa{A}_2 \otimes \spa{A}_3 \otimes \spa{A}$ with $(U_x \otimes I_{\spa{B}}) \ket{\Omega_0} = \ket{\Omega_x}$.
We define the $\CHSH_n$ strategy: \\ 

\begin{compactenum}
\item Alice and Bob share the state $\ket{\Omega_0}$ and receive random $x \in \zo^n$ and $y \in \zo$, respectively. 
\item Alice applies $(U_x)$ such that Alice and Bob share the state $\ket{\Omega_x}$. She measures the space $\spa{A}_2$ in the computational basis to get her outcome $a$.  
\item Bob applies the measurement $\{ M^y_b \}_{b \in \zo^n}$ on his space $\spa{B}$ to determine his outcome $b$. \\ 
\end{compactenum}
 
Conditioned on Alice receiving $x$ and outputting $a$, Bob has the state $\Tr_{\spa{A}} \ketbra{\psi_{a,x\oplus a}}{\psi_{a,x\oplus a}} = \rho_{a,x\oplus a}$. If Bob gets $y = 0$, he must output $b = a$. If Bob gets $y = 1$, he must output $b = a \oplus x$. The probability they win the $\CHSH_n$ game with this strategy is hence equal to $p$. \\ 

\noindent
{\em (3.}$\implies${\em 1.)}. 
Let $\ket{\Omega}_{\spa{A} \spa{B}}$ be the final joint state of the $\OT_p^n$ protocol for honest Alice and Bob. Suppose Alice measures to learn $(z_0, z_1)$ which are distributed uniformly. Let $\rho_{z_0, z_1}$ be Bob's post-measured state. Then, $\{ \rho_{z_0, z_1} : z_0, z_1 \}$ and $\pi$ being the uniform distribution satisfy the hidden XOR condition, since Alice does not abort (both parties are honest), and the protocol is secure. We now describe a procedure to decode each $z_c$, for $c \in \zo$, with probability $p$.

We may assume Bob measures his part of the state $\ket{\Omega}_{\spa{A} \spa{B}}$ (instead of decoding $\rho_{z_0, z_1}$) since it does not matter if Alice measures before or after Bob. Suppose $\ket{\Omega_b}_{\spa{A} \spa{B}}$ is the post-measured joint state when Bob partially measures $\ket{\Omega}_{\spa{A} \spa{B}}$ to obtain his index $b$. Since Bob will not abort and the protocol is secure, we know $b$ is hidden from Alice. Again, by Uhlmann's theorem, Bob can transform $\ket{\Omega_0}$ to $\ket{\Omega_1}$ and vice versa via a unitary acting on $\spa{B}$. Hence Bob can measure $\ket{\Omega}_{\spa{A} \spa{B}}$ to learn $b$, collapse the state to $\ket{\Omega_b}$ and then apply the unitary mapping $\ket{\Omega_{b}}$ to $\ket{\Omega_{c}}$. He then uses the decoding procedure of the $\OT_p^n$ protocol to learn $z_{{c}}$ with probability $p$. \\ 

\noindent 
{\em (4.}$\implies${\em 1.)}.  
Let $\ket{\Omega}_{\spa{A} \spa{B}}$ be the state that Alice and Bob share before receiving $x$ and $y$ in a $\CHSH_n$ game strategy that succeeds with probability $p$. Suppose Alice measures to learn $a$ (conditioned on $x$). Let $\rho_{a,x}$ be Bob's post-measured state which occurs with probability $\pi_{a, x}$. 
We define the necessary states and probabilities by relabelling $a \to x_0$ and $x \oplus a \to x_1$. Then, Bob has no information about $x_0 \oplus x_1 = a \oplus (x \oplus a) = x$ from non-signaling, and 
the average decoding probability for $x_0$ and $x_1$ is $p$. \\ 

Since trivially $(2. \implies 3.)$, we conclude the proof of Theorem~3. 
\end{proof} 

We can also prove an equivalence between quantum encodings of $n$ pairs of bits that hide the XOR of each pair and the $n$-fold repetitions of $\CHSH$ and $\OT$, defined below. 

\begin{definition}[${n}$-fold repetition of oblivious transfer]
A quantum $n$-fold repetition of oblivious transfer protocol with correctness $p$, denoted here as $\OT_p^{\otimes n}$, with cheating probabilities $A_{\OT^{\otimes n}}$ and $B_{\OT^{\otimes n}}$, is defined analogously to an imperfect oblivious string transfer protocol except $b$ is an $n$-bit string (so $z_b$ takes values from each of Alice's strings according to $b$). We say an $\OT_p^{\otimes n}$ protocol is secure if 
Alice can gain no information about the string $b$ (without Bob aborting) and if Bob can gain no information about the string $z_0 \oplus z_1$ (without Alice aborting). 
\end{definition} 

\begin{definition}[${n}$-fold repetition of CHSH]
An $n$-fold repetition of $\CHSH$, denoted $\CHSH^{\otimes n}$, is a game between Alice and Bob where: \\ 

\begin{compactitem}
\item Alice and Bob are allowed to create and share an entangled state $\ket{\psi}$ before the game starts. Once the game starts, there is no further communication between Alice and Bob. 
\item Alice receives a random $x \in \zo^n$ and Bob receives a random $y \in \zo^n$.  
\item Alice outputs $a \in \zo^n$ and Bob outputs ${b \in \zo^n}$.  
\item Alice and Bob win if $a_i \oplus b_i = x_i \cdot y_i$, for all ${i \in \set{1, \ldots, n}}$. \\ 
\end{compactitem} 

The value of the game, $\omega^*(\CHSH^{\otimes n})$, is the maximum probability which Alice and Bob can win. 
\end{definition}  
 
\begin{theorem} \label{thm:equivalence2}
The following four statements are equivalent for every $n \in \mathbb{N}$: \\ 
\begin{compactenum}
\item There is an encoding of $x_0, x_1 \in \zo^n$ that
hides the XOR and $\frac{1}{2^n} \sum_{c \in \zo^n} \Pr[\textup{learn } x_c] = p$, where $x_c \in \zo^n$ is defined as $(x_c)_{i} := (x_{c_i})_i$.  
\item There is a secure, non-interactive $\OT_p^{\otimes n}$ protocol.  
\item There is a secure $\OT_p^{\otimes n}$ protocol.  
\item There is a strategy for winning the game $\CHSH^{\otimes n}$ with probability $p$. 
\end{compactenum}
\end{theorem} 


\section{Applications of equivalences}  
Our equivalences provide new ways of looking at non-local games and cryptographic primitives, through the lens of quantum encodings. Apart from conceptual tools,  we can use the equivalences to prove a number of  results in all areas. 

First, using Theorem~\ref{LearningLemmata} for encodings that hide the XOR with $n=1$ and Theorem~\ref{thm:equivalence2}, we have an alternative proof of the optimality of Tsirelson's bound, $\omega^*(\CHSH) \leq \cos^2(\pi/8)$. 

Using Theorem~\ref{LearningLemmata} for encodings that hide the XOR and Theorem~\ref{thm:equivalence1}, we provide a new upper bound on the value of $\CHSH_n$, $\omega^*(\CHSH_n) \leq \half + \frac{1}{\sqrt{2^{n+1}}}$. 
It is an interesting open question to compute the exact quantum value of this game, especially since it is a simple generalization of the CHSH game for which the quantum value is not known to be implied by {information causality}. 

There is an alternative way of upper bounding the value of this game numerically using semidefinite programming (SDP)~\cite{KRT10}.  
We provide below the values for small $n$. 
We see that the SDP relaxation gives a tighter bound than ours for $n \leq 3$, but the numerical results suggest that our bound outperforms the SDP bound for larger values of $n$. \\ 

\begin{tabular}{|c||c|c|c|c|c|}
\hline
Value & $n=1$ & $n=2$ & $n=3$ & $n=4$ & $n=5$ \\ 
\hline
\hline
Lower Bound & $0.750$ & $0.625$ & $0.562$ & $0.531$ & $0.515$ \\ 
\hline 
Conjectured Value & $0.853$ & $0.750$ & $0.676$ & $0.625$ & $0.588$ \\ 
\hline 
SDP Relaxation & $0.853$ & $0.780$ & $0.743$ & $0.725$ & $0.716$ \\ 
\hline
Our Bound & $1$ & $0.853$ & $0.750$ & $0.676$ & $0.625$ \\ 
\hline
\end{tabular} \quad \\ \quad \\ 

The table above also includes our conjectured optimal value, below. 

\noindent
{\bf Conjecture 1.} 
$\forall n \in \mathbb{N}$, $\omega^*(\CHSH_n) = \half + \half \sqrt{\frac{1}{2^n}}$. 

Similarly, for secure $\OT_p^n$, we have $p \leq \half + \frac{1}{\sqrt{2^{n+1}}}$ (again, for $n=1$, we can get the optimal ${p \leq \cos^2(\pi/8)}$). 

Second, by Theorem~\ref{thm:equivalence2} and the perfect parallel repetition property of CHSH \cite{CSUU08}, i.e., the fact that if Alice and Bob play $n$ games in parallel, the probability of winning all games is exactly $(\cos^2(\pi/8))^n$, we have
for any secure $\OT_p^{\otimes n}$ protocol, ${p \leq (\cos^2(\pi/8))^n}$, which is attainable by using $n$ secure $\OT_{\cos^2(\pi/8)}$ protocols. In other words, secure oblivious transfer admits perfect parallel repetition. 

\section{Robustness of equivalences} 
Similar results can also be obtained in the case of a {\em weighted average decoding probability} defined as \\ 
$q \Pr[\textup{learning } x_0] + (1-q) \Pr[\textup{learning } x_1]$, 
for $q \in [0,1]$. When the XOR is hidden,  and $q = 1/2$, Theorem 1 shows that the above quantity is at most $\cos^2(\pi/8)$. A similar analysis shows that for any $q$, this value is at most
\begin{equation}
\half + \half \sqrt{q^2 + (1-q)^2}. \label{eqn:bound}
\end{equation} 
It is also interesting to see that such a learning relation is related to the CHSH game where Bob gets input $y=0$ with probability $q$ and input $y=1$ with probability $1-q$ while Alice still gets a uniform input. Using a similar method than in Theorem~3, we can show that this game has value at most $\half + \half \sqrt{q^2 + (1-q)^2}$. We can show the optimality of this bound using the semidefinite programming characterization of the bias of XOR games in~\cite{CSUU08}.

\section{Discussion} 
We have provided new relations between the average decoding probability of two bits (or strings) and the probability of decoding their XOR. Moreover, we have shown precise equivalences between quantum encodings, CHSH-type games,  and oblivious transfer, showing that non-locality and cryptographic primitives are often two facets of the same quantum mechanical behaviour. Last, we used our equivalences to prove new results for non-local games and oblivious transfer protocols.

As we have mentioned, it is an open question to compute the quantum value of the game $\CHSH_n$ through semidefinite programming or by proving stronger learning relations. Moreover, we would like to find an information postulate that implies that any theory that abides to it must win this game with exactly the quantum value (similar to information causality for the case of $\CHSH$). \\ 

\nocite{CKS13}
\nocite{Lo97} 
\nocite{CSUU08}
\nocite{CHSH69} 
\nocite{GPSWW06}

\vspace{-0.8cm}
\bibliography{paper}

\appendix 

\vspace{-0.5cm}

\section{Appendix} 

\section{Proof of Equation~(\ref{claim:technical}) and Theorem~\ref{LearningLemmata}} 

Recall Equation~(\ref{claim:technical}) reproduced below,  

\begin{eqnarray*}
 \cos^2(\alpha - \beta) \geq \norm{C D \ket{\psi} }_2^2 + \norm{(1-C)(1-D) \ket{\psi} }_2^2 \nonumber\\
\geq \cos^2(\alpha + \beta). 
\end{eqnarray*} 
 
We first prove the lower bound. Define the following states:  
\[ \ket{X} := \dfrac{C \ket{\psi}}{\norm{C \ket{\psi}}_2}, \quad 
\ket{X'} := \dfrac{(I-C) \ket{\psi}}{\norm{(I-C) \ket{\psi}}_2}, \]  
\[ \ket{Y} := \dfrac{D \ket{\psi}}{\norm{D \ket{\psi}}_2}, \quad 
\ket{Y'} := \dfrac{(I-D) \ket{\psi}}{\norm{(I-D) \ket{\psi}}_2}. \]
We can write $\ket{\psi}$ as \\ 
${\ket{\psi} = \cos(\alpha)\ket{X} + \sin(\alpha)\ket{X'} = \cos(\beta)\ket{Y} + \sin(\beta)\ket{Y'}}$. Since $\ket{X}$ is an eigenvector of $C$, we can write $C = \ketbra{X}{X} + \Pi_C $ and similarly we can write $I-C = \ketbra{X'}{X'} + \Pi_{C'}$, such that 
\begin{align*} 
\inner{\Pi_C}{\kb{X}} 
& = \inner{\Pi_{C'}}{\kb{X}} \\
& = \inner{\Pi_C}{\kb{X'}} \\ 
& = \inner{\Pi_{C'}}{\kb{X'}} \\
& = 0. 
\end{align*} 
We now write $\ket{Y} = \gamma_0 \ket{X} + \gamma_1 \ket{X'} + \gamma_2 \ket{Z}$, 
where $\norm{\ket{Z}}_2 = 1$, $\braket{X}{Z} = \braket{X'}{Z} = 0$, and $|\gamma_0| = \sqrt{x_0}$, $|\gamma_1| = \sqrt{x_1}$, and $|\gamma_2| = \sqrt{x_2}$ for some $x_0, x_1, x_2 \in [0,1]$. Using this expression for $\ket{Y}$, we have 
\begin{align*}
\norm{CD \ket{\psi} }_2^2 
& = \cos^2(\beta)\norm{C \ket{Y}}^2_2 \\ 
& = \cos^2(\beta) \left( x_0 + x_2 \norm{\Pi_C \ket{Z}}^2_2 \right).
\end{align*} 
Since ${\ket{\psi} \! = \! \cos(\alpha)\ket{X} \! + \! \sin(\alpha)\ket{X'} \! = \! \cos(\beta)\ket{Y} + \sin(\beta)\ket{Y'}}$, we can write $\ket{Y'} = \gamma'_0\ket{X} + \gamma'_1\ket{X'} + \gamma'_2\ket{Z}$,  
with $|\gamma'_0| = \sqrt{x'_0}$, $|\gamma'_1| = \sqrt{x'_1}$, and $|\gamma'_2| = \sqrt{x'_2}$ for some $x'_0,x'_1,x'_2 \in [0,1]$.
Using this expression for $\ket{Y'}$, we have 
\begin{align*} 
\norm{(1-C)(1-D) \ket{\psi} }_2^2 
& = \sin^2(\beta)\norm{(1-C) \ket{Y'}}^2_2 \\ 
& = \sin^2(\beta) \left( x'_1 + x'_2 \norm{\Pi_{C'} \ket{Z}}^2 \right).
\end{align*} 
Notice that  
\[ 1 = \norm{C \ket{Z}}^2_2 + \norm{(I-C)\ket{Z}}^2_2 
\comment{= \norm{\ketbra{X}{X} \ket{Z}}^2_2 + \norm{\Pi_C \ket{Z}}^2_2 + \norm{\ketbra{X'}{X'}(\ket{Z})}^2_2 +  \norm{\Pi_{C'} (\ket{Z})}^2_2} 
= \norm{\Pi_C \ket{Z}}^2_2 + \norm{\Pi_{C'} \ket{Z}}^2_2. \]
We define $A := \norm{\Pi_C \ket{Z}}^2_2 = 1 - \norm{\Pi_{C'} \ket{Z}}^2_2$. This yields 
\begin{eqnarray}
& & \norm{CD \ket{\psi} }_2^2 + \norm{(1-C)(1-D) \ket{\psi} }_2^2 \nonumber \\
& = & \cos^2(\beta) \left( x_0 + x_2 \norm{\Pi_C \ket{Z}}^2_2 \right) \nonumber \\ 
& & + \sin^2(\beta) \left( x'_1 + x'_2 \norm{\Pi_{C'} \ket{Z}}^2_2 \right) \nonumber \\
& = & \cos^2(\beta) \left( x_0 + x_2 A \right) + \sin^2(\beta) \left( x'_1 + x'_2 (1-A) \right) \nonumber \\
& = & \cos^2(\beta)x_0 + \sin^2(\beta) \left( x'_1 + x'_2 \right) \nonumber \\ 
& & + A \left( \cos^2(\beta)x_2 - \sin^2(\beta)x'_2 \right) \nonumber \\
& = & \cos^2(\beta)x_0 + \sin^2(\beta) \left( 1 - x'_0 \right) \nonumber \\ 
& & + A \left( \cos^2(\beta)x_2 - \sin^2(\beta) x'_2 \right). \label{eqn:2}
\end{eqnarray} 
Define $A(\rho, \sigma) := \arccos F(\rho, \sigma)$ to be the angle between two states $\rho$ and $\sigma$, which is a metric (see p. $413$ in \cite{NC00}).
Since $\braket{Y}{Y'} = 0$, we have 
\[ A(\ket{Y'},\ket{X}) \ge \pi/2 - A(\ket{X},\ket{Y}). \] 
This implies that 
\begin{align*} 
\sqrt{x'_0} 
& = \cos \left( \arccos |\braket{Y'}{X}| \right) \\ 
& \leq \cos \left( \pi/2 - \arccos \sqrt{x_0} \right) \\ 
& = \sin \left( \arccos \sqrt{x_0} \right) \\
& = \sqrt{1-x_0}. 
\end{align*} 
This yields $x'_0 \le 1 - x_0$.
In addition, notice that $\braket{\psi}{Z} = 0$, which implies that 
\begin{eqnarray*} 
& & \bra{Z} \left( \cos(\beta)\ket{Y} + \sin(\beta)\ket{Y'} \right)= 0 \\ 
& \iff & \cos^2(\beta)|\braket{Z}{Y}|^2 = \sin^2(\beta)|\braket{Z}{Y'}|^2 \\ 
& \iff & \cos^2(\beta)x_2 = \sin^2(\beta) x'_2.
\end{eqnarray*}
This gives us the bound,  
\begin{equation} 
\norm{CD \ket{\psi} }_2^2 + \norm{(1-C)(1-D) \ket{\psi} }_2^2 \ge x_0. 
\end{equation}
To conclude, we have 
\begin{align*} 
\arccos(\sqrt{x_0}) 
& = A(\ket{X},\ket{Y}) \\ 
& \leq A(\ket{X},\ket{\psi}) + A(\ket{\psi},\ket{Y}) \\
& \leq \alpha + \beta, 
\end{align*} 
yielding $x_0 \ge \cos^2(\alpha + \beta)$ which concludes the proof of the lower bound.  

For the upper bound, we have $x'_0 \leq 1 - x_0$ and $\cos^2(\beta)x_2 = \sin^2(\beta) x'_2$, hence,
\[ 
\norm{CD \ket{\psi} }_2^2 + \norm{(1-C)(1-D) \ket{\psi} }_2^2 
\leq 1 - x'_0, \] 
from (\ref{eqn:2}). We now show $1 - x'_0 \leq \cos^2(\beta - \alpha)$. Since $\sqrt{x'_0} = |\braket{Y'}{X}|$,  we have 
\begin{align*}
\arccos \left( \sqrt{x'_0} \right) 
& = A(\ket{Y'}, \ket{X}) \\ 
& \leq A(\ket{X}, \ket{\psi}) + A(\ket{Y'}, \ket{\psi}) \\ 
& = \pi/2 - (\beta - \alpha). 
\end{align*} 
so $\sqrt{x'_0} \geq \cos(\pi/2 - (\beta - \alpha)) = \sin(\beta - \alpha)$ 
implying $1 - x'_0 \leq \cos^2(\beta - \alpha)$, as desired. $\qed$ \\

\paragraph{Proof of Theorem 1.} 
The proof of the first statement in the theorem relies on the following decoding strategy: First, we apply the decoding procedure for learning the first bit and then we apply the second decoding procedure on the post-measurement state. The probability of decoding the XOR is the probability that both decoding procedures succeed (give correct guesses for each bit) or they both fail (give incorrect guesses for each bit). 

We prove the theorem using the following (equivalent) setting. We suppose two parties, Alice and Bob, share a joint pure state $\ket{\Omega}_{\spa{A}\spa{B}}$ such that Alice performs a projective measurement ${M = \{ M_{x_0,x_1} \}_{x_0, x_1 \in \zo}}$ on $\spa{A}$ to determine $x_0$ and $x_1$ and the post-measured state is Bob's encoding of $x_0$ and $x_1$. Let $p_i$ be the maximum probability that Bob can learn bit $x_i$, for $i \in \{ 0, 1 \}$. We note that without loss of generality, Bob can perform a projective measurement to guess the value of $x_i$ with maximum probability~\cite{NC00}. Let $P = \{ P_0, P_1 \}$ be Bob's projective measurement that allows him to guess $x_0$ with probability $p_0 = \cos^2(\alpha) \geq \half$ and $Q = \{ Q_0, Q_1 \}$ be Bob's projective measurement that allows him to guess $x_1$ with probability $p_1 = \cos^2(\beta) \geq \half$ (these measurements are on $\spa{B}$ only). Consider the following projections (on $\spa{A} \otimes \spa{B}$):
\[ C = \sum_{x_0,x_1} M_{x_0,x_1} \otimes P_{x_0} \quad \text{ and } \quad D = \sum_{x_0,x_1} M_{x_0,x_1} \otimes Q_{x_1}. \]
$C$ (resp. $D$) is the projection on the subspace where Bob guesses correctly $x_0$ (resp. $x_1$) after applying $P$ (resp. $Q$). Consider the strategy where Bob applies the two measurements $P$ and $Q$ one after the other to learn $(x_0, x_1)$, from which he can calculate $x_0 \oplus x_1$. If both guesses are correct or if both guesses are incorrect then his guess for $x_0 \oplus x_1$ is correct.

Let Bob perform the following projective measurement to learn both bits: 
\[ R = \{ R_{x_0, x_1} := Q_{x_1} P_{x_0} Q_{x_1} \}_{x_0, x_1 \in \zo}. \] 
The measurement where Bob guesses both bits correctly when applying $R$ is 
\[ E = \sum_{x_0,x_1} M_{x_0,x_1} \otimes R_{x_0, x_1} = DCD, \] 
with outcome probability $\bra{\Omega} E \ket{\Omega} = \norm{CD \ket{\Omega}}_2^2$.
The measurement where Bob guesses both bits incorrectly when applying $R$ is 
\[ F = \sum_{x_0,x_1} M_{x_0,x_1} \otimes R_{\bar{x_0}, \bar{x_1}} = (I-D)(I-C)(I-D). \] 
The probability of this measurement outcome is $\bra{\Omega} F \ket{\Omega} = \norm{(I-C)(I-D) \ket{\Omega}}_2^2$. 
With this strategy, Bob can guess $x_0 \oplus x_1$ with probability 
\[ || CD \ket{\Omega}||^2_2 + ||(I-C)(I-D) \ket{\Omega}||_2^2 \geq \cos^2(\alpha + \beta) \]
by (1). Note that 
\[ c := \frac{p_0 + p_1}{2} = \frac{\cos^2(\alpha) + \cos^2(\beta)}{2} \ge \half \] 
and for such values of $\alpha, \beta$, we have $\cos(\alpha + \beta) \geq \cos^2(\alpha) + \cos^2(\beta) - 1$. Therefore,  
\[ \Pr[\textup{Bob can learn } x_0 \oplus x_1] \geq \cos^2(\alpha + \beta) \geq (2c - 1)^2. \] 

For the second statement, ideally, we would like to extend our proof approach from bits to strings, but unfortunately this statement is not true anymore if $x_0$ and $x_1$ are strings. Instead, the analysis in~\cite{CKS13} can be generalized to strings to show 
\[ \Pr[\textup{learning } (x_0,x_1)] \geq \left( \frac{\cos^2(\alpha) + \cos^2(\beta)}{2} \right) \cos^2(\alpha + \beta). \] 
If $c \geq 1/2$, then by the same reasoning as above, we have 
$\Pr[\textup{learning } (x_0,x_1)] \geq c(2c-1)^2.$ The statement about the XOR follows directly from the above statement. $\qed$ 

\section{Proofs of the security bounds for oblivious transfer protocols}

We now provide proofs of the lower bounds of $59.9\%$ and $58.52\%$ for any oblivious transfer and oblivious string transfer protocol, respectively, with $p=1$, by relating them to bit commitment. A quantum bit commitment protocol, denoted $\BC$, is an interactive protocol with no inputs, between Alice and Bob, with two phases: \\ 

\begin{compactitem}
\item Commit phase: Bob chooses a random $b$ and interacts with Alice to commit to $b$. 
\item Reveal phase: Alice and Bob interact to reveal $b$ to Alice. 
\item If the parties are honest, Alice accepts the value of $b$. \\
\end{compactitem}

We say that the bit commitment protocol has cheating probabilities $A_{\BC}$ and $B_{\BC}$ where \\ 

\begin{compactitem}
\item $B_{\BC} \! := \! \max \! \left\{ \! \dsum_{b \in \zo} \! \frac{1}{2}  \Pr[\textup{Alice accepts outcome } b] \! \right\}$, 
\item $A_{\BC} := \Pr[\textup{Alice can learn } b \textup{ after commit phase}]$. \\ 
\end{compactitem}

We present a bit commitment protocol based on oblivious string transfer~\cite{CKS13}. \\ 

\vspace{-0.25cm}

\begin{compactenum}
\item Commit phase: Alice and Bob perform the $\OT_1^n$ protocol such that Alice gets the output ${(z_0, z_1) \in \zo^n \times \zo^n}$ and Bob gets the output $(b, w) \in \zo \times \zo^n$. Here, $b$ is the committed bit. 
\item Reveal phase: If no one aborted, then Bob sends $(b, w)$ to Alice.  
\item If $(b, w)$ from Bob is inconsistent with $(z_0, z_1)$ then Alice aborts. Otherwise, she accepts $b$ as the committed bit. \\ 
\end{compactenum}

\vspace{-0.25cm}

Let $A_{\OT^n}$ denote the probability Alice can learn $b$ in the $\OT_1^n$ protocol without Bob aborting. Clearly we have $A_{\OT^n} = A_{\BC}$. 

Let $B_{\OT^n}$ denote the probability Bob can learn ${z_0 \oplus z_1}$ in the $\OT_1^n$ protocol without Alice aborting. Notice that Bob must send $(c, z_c)$ if he wants to reveal $c$ in the BC protocol. Therefore, by letting $q$ be the probability the $\OT_p^n$ is not aborted by Alice using Bob's optimal bit commitment strategy, we have $B_{\BC} = qc$, where \\ 
$c = \dfrac{1}{2} \dsum_{b \in \zo} \Pr[\textup{Bob learns } z_b | \textup{Alice did not abort } \OT_1^n ]$. 
From Theorem~1, we know that Bob has a strategy to learn $(z_0, z_1)$ with probability, 
\vspace{-0.1cm}
\[ B_{\OT^n} \geq q c (2c-1)^2, \] 
noting that $B_{\BC} \geq 1/2 \implies c \geq 1/2$. 

We now use the lower bound for bit commitment~\cite{CK11}, which states that  there is a parameter $t \in [0,1]$ such that 
\[ B_{\BC} \geq \left( 1 - \left( 1 - \dfrac{1}{\sqrt 2} \right) t \right)^2  \quad \text{ and } \quad A_{\BC} \geq \half + \dfrac{t}{2} . \] 

The above bound yields the lower bound $\max \{ A_{\OT^n}, B_{\OT^n} \} \geq 0.5852$, which is independent of $n$. If $n=1$, we can use the stronger bound in Theorem~1 to get 
\vspace{-0.2cm}
\[ B_{\OT} \geq q (2c-1)^2, \] 
improving the lower bound to the desired value $\max \{ A_{\OT}, B_{\OT} \} \geq 0.599$. $\qed$ 

\end{document}